\documentclass[11pt]{article}
\usepackage{latexsym,amssymb,amsmath,cases}
\usepackage{algorithm}
\usepackage{algpseudocode}
\usepackage{authblk}
\usepackage{amsthm}
\usepackage{graphicx}
\usepackage{xcolor}
\newtheorem{observation}{Observation}

\newtheorem{lemma}{Lemma}
\newtheorem{problem}{Problem}
\newtheorem{definition}{Definition}
\newtheorem{theorem}{Theorem}
\newcommand{\IN}{ proper }
\newcommand{\MON}{ monotone }
\DeclareMathOperator{\ret}{ret}
\begin{document}
\title{The En Route Truck-Drone Delivery Problem}

\author[1]{Danny Krizanc}
\author[2]{Lata Narayanan} 
\author[2]{Jaroslav Opatrny} 
\author[2]{Denis Pankratov}
\affil[1]{Department of Mathematics \& Comp. Sci., Wesleyan University, Middletown, USA}
\affil[2]{Department of CSSE, Concordia University, Montreal, Canada}
\maketitle

\begin{abstract}
We study the truck-drone cooperative delivery problem in a setting where a single truck carrying a drone  travels at constant speed on a straight-line trajectory/street. Delivery to clients located in the plane and not on the truck's trajectory is performed by the drone, which has limited carrying capacity and flying range, and whose battery can be recharged when on the truck. We show that the problem of maximizing the number of deliveries is strongly NP-hard even in this simple setting. We present a 2-approximation algorithm for the problem, and an optimal algorithm for a non-trivial family of instances.  
\end{abstract}
\section{Introduction}

The use of unmanned aerial vehicles or drones for last-mile delivery in the logistics industry has received considerable attention in business and academic communities, see for example \cite{dr1,dr2,dr3,macrina2020drone}.  
Drones have been shown in a recent analysis  \cite{Raghu2023} to have significantly less lifecycle costs, and faster delivery time compared to diesel or electric trucks in urban, suburban, and rural settings, and have less harmful emissions compared to diesel trucks. 
The potential applications where drone delivery could make a big impact include
contactless delivery, return of unsatisfactory goods, rural or hard-to-access delivery and delivery
in disaster relief scenarios.


In this paper we consider a system in which the delivery of physical items to clients located in the plane is done by {\em two cooperating mobile agents} having different but complementary properties. 
The first mobile agent, called the {\em drone} can move in any direction but it can travel only a limited distance, called its {\em flying range}, before it needs to recharge its battery. Furthermore, it has limited carrying capacity. 
The second mobile agent, called the {\em truck} can travel only along a fixed trajectory, called a {\em street} but its battery/fuel is not only sufficient to follow the street as long as necessary, but it  is also equipped with a charging facility where the drone can recharge whenever it reaches the truck. Furthermore, it can carry all items that are to be delivered to the clients. 

The delivery of items to clients is done as follows. 
All items to be delivered are preloaded on the truck at the warehouse. The truck then moves along the street at a fixed speed and it delivers items to any client  who is located on its trajectory. The delivery of an item to a client  who  is not located on the trajectory of the truck must be carried out by the drone. At an appropriate time, the drone  flies from the truck with the item to be delivered to the given client, drops the item there, and then  flies back to  the still-moving truck. There it can recharge, pick up  another item, and  make the next  delivery, and so on. Clearly the same set-up can also be used to pick up items rather than deliver them. For ease of exposition, we always talk about item delivery in this paper. 
 
Given a set of delivery locations and the parameters of the agents, i.e., the trajectory and the speed of the truck, the flying range of the drone and its speed,  we want to compute a feasible schedule of deliveries that {\em maximizes the number of deliveries} made.
Such a schedule specifies the order in which the deliveries to clients are 
done by the drone, and for each delivery it gives  the time the drone leaves the truck. Clearly, to be feasible, the schedule should ensure that for each delivery, the drone can fly to the delivery location and back to the still-moving truck while having travelled distance at most its flying range, and arrive at the truck in time to start its next delivery.    


\subsection{Related work}

The algorithmic study of truck-drone cooperative delivery problems was initiated by Murray and Chu \cite{murray2015flying} and Mathew et al. \cite{mathew2015optimal} where the problem of a single truck being helped by a single drone to deliver packages to customers is studied. Since then there has been a great deal of work (Murray and Chu's paper has received more than 1000 citations) on different versions of what is variously
referred to as Truck-Drone Cooperative Delivery, Drone-Aided Delivery or Last-Mile Delivery problems. Variations considered include multiple trucks, multiple drones, drone-only delivery, mixed truck-drone delivery, etc. We refer the reader to recent surveys for more details \cite{freitas2023exact,dr2,macrina2020drone,dr3,zhang2023review}.

In the above work, the problem is most often modelled using a weighted directed graph with customers as nodes, streets and
drone flight paths as edges, etc. Under these circumstances the problems become versions of the Travelling Salesperson Problem or the Vehicle Routing Problem. As such they are all easily seen to be NP-hard in general and are solved by adapting known exact (e.g., Mixed Integer Linear Programming) or heuristic (e.g., greedy) techniques. For specialized domains some variants can be shown to be polynomial time, e.g. on trees \cite{erlebach2022package}. 

  In most of the previous research it is assumed that the points at which a truck and drone can rendezvous are part of the input (e.g., customer locations, depots) and that the truck or drone stops at the rendezvous point to wait for the other to arrive. More recent work \cite{khanda2022drone,li2022truck,masone2022multivisit,thomas2023collaborative} has focused on the case where the rendezvous can occur ``en route'' as the truck is moving and the rendezvous points are to be determined by the algorithm, as is the case with our study. In these papers, the problems studied are again generalized versions of TSP or VRP and are attacked via adaptations of known exact or heuristic techniques. Here we restrict ourselves to the simplest version of the problem with one truck and one drone, where the truck travels at a constant speed along a single street. Surprisingly, even in this case, as shown in Section~\ref{sec:np}, the problem is strongly NP-hard. 

All of the above work is concentrated on minimizing either the total delivery time or total energy requirements (or some combination of both) to deliver all of the packages to all of the customers. To the best of our knowledge we are the first to consider the problem of maximizing the number of clients  that are satisfied in the en route model. 


\subsection{Our Truck-Drone Model}
We define the truck-drone delivery problem more formally as follows. We assume that the delivery points as well as the trajectories  of the truck and the drone, are   set in the 2-dimensional Cartesian plane. Without loss of generality, we assume the warehouse is located at $[0,0]$, and the truck starts fully loaded with all items to be delivered at the warehouse at time 0, and subsequently moves right on the $x$-axis with constant speed 1. Note that this allows us to
measure the elapsed time by the distance of the truck from the origin. 

The speed of the drone is denoted by $v$ and it is assumed that $v$ is a constant that is greater than 1.  The  {\em flying range} of the drone is given by the value $R$, and is defined as the maximum distance that the drone can fly on a full battery without needing to be recharged. We assume that the time to recharge the drone's battery,  and to pick up  an item from the truck, or to drop off an item at its delivery location are negligible compared to the delivery times, and thus are equal to $0$. Therefore, any time the drone leaves the truck it can fly  its full range $R$ before returning to the truck. 

We are given  a multi-set $D=\{d_1,d_2,\ldots,d_n\}$ of delivery points in the plane where the deliveries are to be made. 
The truck delivers any item whose delivery point is located 
is on its trajectory, we assume that this can be done with negligible delay. Thus we assume 
below that none of the points  in $D$ is located on the trajectory of the truck,
{\em i.e., on the positive $x$-axis.}



We now define a feasible delivery schedule for the truck-drone delivery problem.

\begin{definition}
Given an instance  $I=(v,R,D)$ of the truck-drone problem, where $D=\{d_1,d_2,\ldots,d_n\}$, we define a {\em schedule} $\mathcal{S}_I$ to be an {\em ordered list} of delivery points to which deliveries  are made, and the start time  of each delivery, i.e.,
 $$\mathcal{S}_I=((d_{i_1},s_1),(d_{i_2},s_2),\ldots, (d_{i_m},s_m)), m\leq n$$ where  $m$ is called the {\em length} of the schedule and for $1\leq j\leq m$ 
the drone makes a delivery to $d_{i_j}$ by leaving the truck at point $[s_j,0]$.  
   The  schedule is   {\em feasible}, if 
 $s_1\geq 0$, and for each $j$, $1\leq j\leq m-1$, the drone can reach
$d_{i_j}$ when leaving the truck at position $[s_j,0]$ and return to the truck
at or before  $[s_{j+1},0]$.
\end{definition} 

Schedule  $\mathcal{S}_I$
 is called {\em optimal} if  there is no schedule  that is longer than  $\mathcal{S}_I$, that is, makes more deliveries than $\mathcal{S}_I$. 
 
Given an instance $I=(v,R,D)$ of the truck-drone problem, where  $v$ and $R$ are the speed and the range of the drone respectively, and $D$ is the set of delivery points,  the goal of the truck-drone delivery problem is to find an optimal delivery schedule.

\begin{figure}[ht]
  \centering
\includegraphics[width=0.8\textwidth]{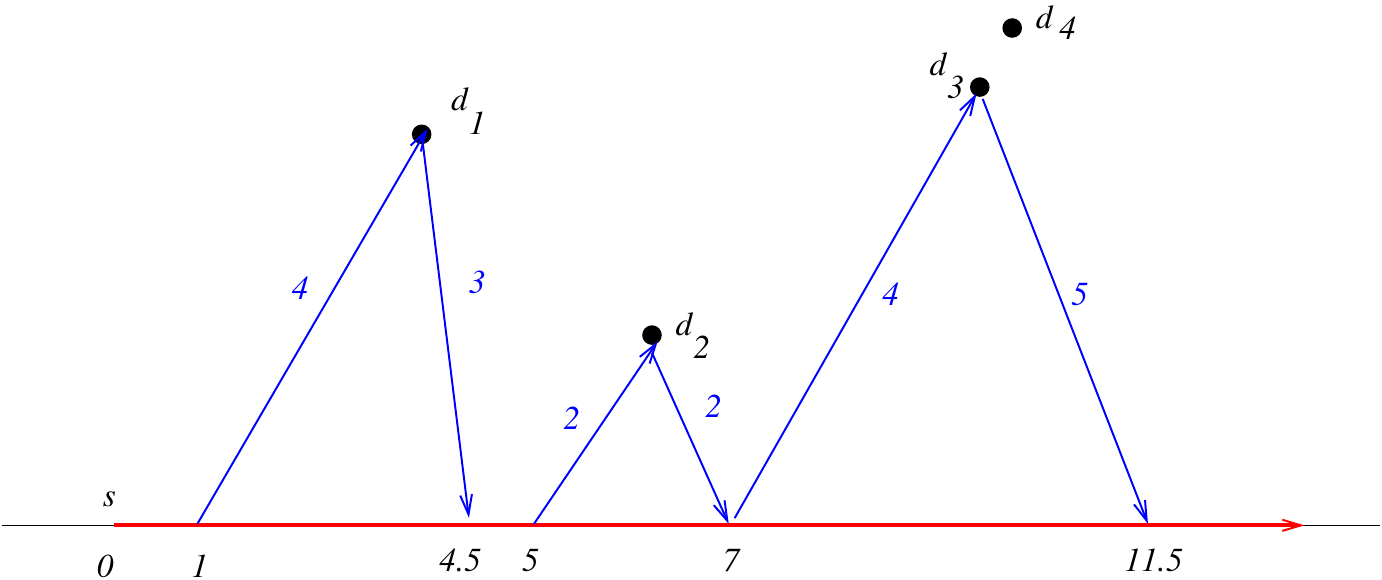}
\caption{ Instance $I=(2,10,\{d_1,d_2,d_3,d_4\})$, and its  schedule $\mathcal{S}_I=((1,d_1),(5,d_2),(7,d_3))$.
The trajectory of the drone is in blue, that of the truck in red. 
The blue numbers give the distances, the black numbers show the time sequence}.
  \label{fig:sch}
\end{figure}

Figure \ref{fig:sch} shows an example of a truck-drone problem and of a feasible schedule.    

\subsection{Our results}

In Section~\ref{sec:np}, we show that even for the ostensibly simple case of a single truck travelling on a straight line, and a single drone, the truck-drone delivery problem is  strongly NP-hard. In particular, we show that 
given an instance $I$ of the truck-drone problem and an integer $k$, it
is  strongly NP-hard \cite{NPbook} to decide whether there is a schedule $\mathcal{S}_I$ 
of length $k$.  

In Section~\ref{greedy-alg}, we  describe a greedy algorithm $\mathcal{A}_g$ and show that it computes a 2-approximation of an optimal schedule in $O(n^2)$ time. The factor of $2$ is shown to be tight for this algorithm. 
Finally, in Section~\ref{sec:optimal}, we 
define a {\em \IN} family of instances. Roughly speaking, in such instances, the delivery points do not  have the same  or ``nearly'' the same $x$-coordinates, where ``nearly'' depends on the difference in their $y$-coordinates.  In particular, the greater the difference in the $y$-coordinates of the points, the greater is the difference in their $x$-coordinates in \IN instances. 
We then give an $O(n^3)$  algorithm  that calculates an optimal schedule
for any \IN instance. 


\section{Preliminary Observations}
\label{prelims}


We say that a point $d=[x,y]$ is {\em reachable} by the drone 
from  position   $ [s,0]$  if the drone can leave the truck at $[s,0]$, fly to point $d$ and fly back to the truck with the total distance travelled at most its flying range $R$.  
First we examine some geometric properties of points in the plane that are reachable from $[s, 0]$ by the drone flying with speed $v$ and having flying range $R$. 

Suppose the drone leaves the truck at position $[s,0]$, makes a delivery at 
$d=[x,y]$ and returns to the truck 
using its {\em full range} $R$. 
To fly range $R$  the drone needs time $t = R/v$ and at that time the truck is at position $[s+R/v,0]$.
\begin{figure}[ht]
\centering
\includegraphics[scale=0.8]{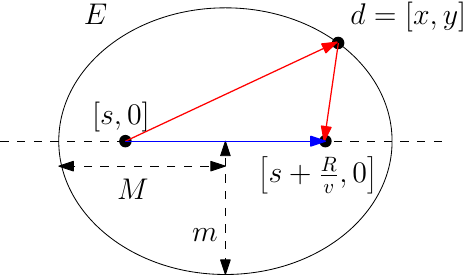}
\caption{In  Ellipse $E$ shown above, the speed of the drone is not much higher than that of the truck. When the speed of the drone increases, the distance between the foci decreases, and the ellipse becomes closer to  a circle.}
\label{fig:ellipse}
\end{figure}
Therefore, the drone can make a delivery at  
point $d=[x,y]$ if the total distance it flew satisfies the equation 
$$|[s,0],[x,y]| +|[x,y], [s+R/v,0]| =R$$
Clearly all such points $d$ reachable by the drone from  $[s,0]$ using its full flying range lie on {\em ellipse} $E$ (see Figure \ref{fig:ellipse}) with 
left focus 
$[s,0]$ and right focus $[s+R/v,0]$. Furthermore, the {\em major radius}, i.e. the length of the 
{\em semi-major axis} of the ellipse is $M=\frac R 2$, and {\em minor radius}, i.e. the length of its 
{\em semi-minor axis} is $m=\frac{R}{2v}\sqrt{v^2 -1}$.
Next, considering also the delivery points that can be reached by the drone by flying distance $<R$, we conclude that all  points reachable from $[s, 0]$ by the drone {\em within} its flying range are located {\em on or inside} the ellipse $E$.

Assuming that the ellipse $E$ is  centered at $[0,0]$, its 
left focus $[s, 0]= [-\frac{R}{2v},0]$, and its right focus is $[\frac{R}{2v},0]$, and $M$, $m$ are the major, minor radii as specified above. The equation of the ellipse is:
\begin{equation}
\frac{x^2}{M^2} +\frac{y^2}{m^2} = 1
\label{ellipse}
\end{equation}
Clearly, delivery to  point $d=[x,y]$ is {\em feasible} 
only if  $-m\leq y\leq m$, i.e., all delivery points should be located in a band of width $2m$ centered along  the $x$ -axis.  


 
Assume a delivery point $d$ is on the right half  of ellipse $E$, and 
the drone makes a delivery to $d$ starting from the truck at point $[s',0]$ between the foci of the ellipse $E$. 
Since the distance from $[s',0]$ to $d$ is shorter than the distance from the left focus $[s, 0]$ of $E$ to $d$, the drone can reach the delivery point $d$, flying for  distance  $<R$.
However,  the drone when leaving the truck at point $[s,0]$ 
arrives at $d$ {\em earlier} than  when staying on the truck and leaving for 
$d$ only later at point  $[s',0]$, and therefore it also returns to the 
truck earlier.  Thus when using flying distance less than $R$ the drone returns to the truck  {\em later} as shown in Figure \ref{fig:ellipse2}.
\begin{figure}[ht]
\centering
\includegraphics[scale=0.6]{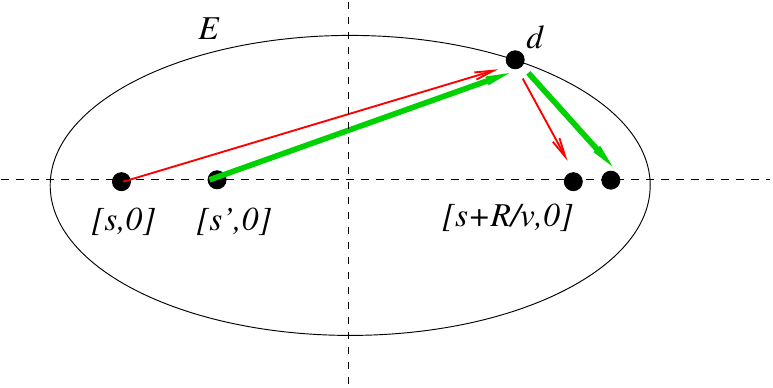}
\caption{ The red lines show the delivery with the full range $R$, the green lines show the delivery with range less than $R$.} 
  \label{fig:ellipse2}
\end{figure}

This leads to the next observation: 
\begin{observation}
 Consider a delivery point $d$ in the right  half  of the ellipse $E$.
  To make a delivery to $d$  flying less that the full range $R$, the drone must start the delivery  at a point to the right  of the left focus of $E$ and  the drone returns to the truck to the right  of the right focus of $E$. Starting points for the drone to the left  of the left  focus are not feasible. 
\label{obs:2}
 \end{observation}

A symmetric observation holds about delivery points on the left half of  $E$.  

We now  determine for each delivery point an interval on the trajectory of the truck describing feasible departure points for the drone to make a delivery to point $d$. 
Given a delivery point $d$, let  $E_1$ and $E_2$ be the ellipses with major radius $M$ and minor radius $m$, such that their  foci are located on the $x$-axis,  with $E_1$ containing
  $d$ on its right half, while $E_2$ contains $d$ on its left half. Let
  $f_{i1}, f_{i2}$ be the foci of $E_i$ for $i \in \{1, 2 \}$ (see Figure \ref{fig:ellipse3}). The following observation now follows from  Observation \ref{obs:2} above.
\begin{observation}
Focus $f_{11}$ is the point of the {\em earliest start} for a delivery to $d$, and  focus $f_{12}$ is the point of the {\em earliest return} to the truck from a delivery to $d$.  
Focus $f_{21}$ is the point of the {\em latest start} for a delivery to $d$ that can meet the truck,  
and Focus $f_{22}$ is the  {\em latest return} to the truck from any delivery
to $d$. Feasible start points for delivery to $d$ lie between $f_{11}$ and $f_{21}$, with the corresponding return to the truck occurring between $f_{12}$ and $f_{22}$.
 \label{obs:focus}
\end{observation}

\begin{figure}[ht]
\centering
\includegraphics[scale=0.8]{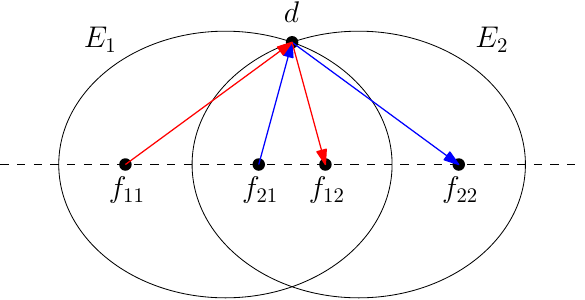}
\caption{ The red lines show the earliest delivery to $d$, the blue lines show the latest delivery to $d$. A delivery to $d$ could be scheduled to start at
a point between $f_{11}$ and $f_{21}$.}
  \label{fig:ellipse3}
\end{figure}
In the rest of this paper, given a delivery point $d$ we denote its earliest  start  time 
as $es(d)$  and the corresponding earliest return as $er(d)$, the latest start time of $d$ as $ls(d)$, and
the corresponding latest return back to the truck as $lr(d)$, 

Notice that for any delivery point $d$ we have 
$$er(d) -es(d)=lr(d)-ls(d)=R/v$$ the distance between the foci of $E$.\\

Given a point $d=[x,y]$, we can calculate the values $es(d)$, $ls(d)$ as follows. Imagine a horizontal line passing through $d$. It intersects the ellipse $E$ centered at $0$ at two points $[-x', y]$ and $[x',y]$. According to Equation~\eqref{ellipse}, we have $(x')^2/M^2 + y^2/m = 1$. Therefore, $x'= M\sqrt{(1 -\frac{y^2}{m^2})}$. Now, imagine sliding the ellipse $E$ along the $x$-axis. When $E$ touches $d$ for the first time, we obtain $E_1$ having travelled distance $x-x'$. Similarly, when $E$ touches $d$ for the last time, we obtain $E_2$ having travelled distance $x+x'$. Thus, we have:
\begin{observation}
\label{obs:5} For $d=[x,y]$

$es(d) = x-\frac{R}{2v}-x'$, and $er(d) = es(d)+R/v,$

$ls(d) = x-\frac{R}{2v}+x'$, and $lr(d) = ls(d) +R/v.$
\end{observation}

The next lemma gives the return point of the drone to the truck after a delivery to a delivery point $d = [x, y]$, starting from the truck at a position $[s, 0]$.

\begin{lemma}
\label{lem:return_time_from_starting_pos}
Suppose we wish to make a delivery to a delivery point $d=[x,y]$ using the drone, starting from the truck at position $[s,0]$, and returning to the truck at position $[\ret, 0]$. 

\begin{enumerate}
\item If $es(d) \leq s\leq ls(d)$.
 
\begin{equation} 
\ret=\ret(s,d,v):=s+\frac{s+av-x+\sqrt{b^2 -s(v^2-1)(b+s+av-x)}}{v^2-1}
\label{return}
\end{equation}
where  $a=\sqrt{y^2+(s-x)^2}$,  $b=sv^2+av-x$.
\item If $s < es(d)$, then $\ret(s, d, v) = er(d)$.
\item If $s > ls(d)$, then delivery is impossible, thus we set $\ret(s,d,v) = \infty$.

\end{enumerate}
\end{lemma}
\begin{proof}
To see (1), observe that the  total  distance travelled by the drone is  
$d_1=|[s,0],[x,y]| +|[x,y],[\ret,0]|=
a+\sqrt{(\ret-x)^2+y^2}$, which the drone travels in time $d_1/v$. At the same time the truck travels the distance $d_2=\ret-s$. 
Thus we have the equation
\begin{align*}a+\sqrt{(x-\ret)^2+y^2}&=v(\ret-s)\\
(x-\ret)^2+y^2&=\left( v(\ret-s)-a\right)^2\\
x^2 -2x\ret +\ret^2 +y^2&=v^2(\ret^2-2s\ret+s^2) -2av(\ret-s)  +y^2+(s^2-2sx +x^2)\\
\ret^2 -2x\ret &=v^2\ret^2 -2v^2s\ret -2av\ret+v^2s^2+2avs +s^2-2sx\\
0&=(v^2-1)\ret^2-2(sv^2+av-x)\ret +s(v^2s+s+2av-2x)\\
0&=(v^2-1)\ret^2-2b\ret +s(b+s+av-x),
\end{align*}
and by solving the quadratic equation for $\ret>s$ we have 
\begin{align*}\ret&=\frac{b +\sqrt{b^2 -s(v^2-1)(b+s+av-x)}}{v^2-1}\\
&=s+\frac{s+ av-x+\sqrt{b^2 -s(v^2-1)(b+s+av-x)}}{v^2-1},
\end{align*}
as needed. 

For (2), note that if  $s < es(d)$ then  the drone remains on the truck until 
position $[es(d),0]$ is reached and then it starts  a delivery from position
$[es(d),0]$, since by Observation \ref{obs:2}, this gives the earliest time the drone can start from the truck for a delivery to  point $d$.  Thus for any such $s$ the drone returns to the truck  at position  $[er(d), 0]$.

Finally, (3) follows from Observation \ref{obs:focus}.
\end{proof}

For $s$ where   $es(d) \leq s\leq ls(d)$ and a delivery point $d=[x,y]$,   
we call $\ret(s,d,v)-s$ the {\em round-trip flight time} to $d$ from $[s,0]$. 
It can be seen from Formula \ref{return} that   
the round-trip flight time is not a linear function in $s$, which makes a calculation of a schedule for a given instance of the truck-drone problem 
more complicated. 
\begin{observation} 
\label{obs:profile}
For a delivery point $d=[x,y]$ and a point $[s,0]$ 
 between   $es(d)$ and $ ls(d)$,  
the round-trip flight time  $\ret(s,d,v)-s$ reaches the maximal value $R/v$ 
at $s= es(d)$, it decreases until $s=x(1-1/\sqrt{v^2-1})$  and then increases 
until $s= ls(d)$ where   it again reaches the maximal value $R/v$. 
\end{observation}

\begin{lemma}
Let $d=[x,y]$ and $d'=[x',y']$ be two delivery points, and suppose there are valid drone trajectories from $[s,0]$ to $d$ returning at $[r,0]$ and from $[s',0]$ to $d'$ returning at $[r',0]$. If $s' < s < r \leq r'$, then there is also a valid drone trajectory from $[s',0]$ to $d$ returning at a point before $[r,0]$. 
\label{lem:valid-trajectories}
\end{lemma}
\begin{proof}
Let $R_1$ be the length of the drone trajectory from $[s',0]$ to $d'$ and then to $[r',0]$, and similarly, let $R_2$ be the length of the drone trajectory from $[s,0]$ to $d$ and then to $[r,0]$. Then $R_1/v$ and $R_2/v$ respectively are the distances from $[s',0]$ to $[r',0]$ and from $[s,0]$ to $[r,0]$. Since $s' < s < r \leq r'$, it follows that $R_2 < R_1$. Now consider the ellipse $E_1$ with parameters $(R_1, v)$ with $[s',0]$ as its left focus. Then $d'$ is on the right half of $E_1$, and $[r',0]$ must be its right focus. Similarly, let $E_2$ be the ellipse with parameters $(R_2, v)$ with $[s,0]$ and $[r,0]$ as its left and right foci respectively, and with $d$ on the right half of the ellipse. Since $s' < s < r \leq r'$, the ellipse $E_2$ is completely contained in $E_1$, and the point $d$ is in the interior of the ellipse $E_1$. It follows that there is a valid drone trajectory to $d$ starting at $[s',0]$. Furthermore, since the drone reaches $d$ earlier if it starts at $[s',0]$ than if it stayed on the truck until $[s,0]$ and then flew to $d$, it must also return to the truck earlier than $[r,0]$. 
\end{proof}

In the truck-drone instance that we use in the proof of strong NP-hardness in Section~\ref{sec:np},  many of the delivery points are located on the $y$ axis. For these points we can simplify 
the expression used to define function $\ret(s, d,v)-s$, and this simplified expression  is used to obtain upper and lower  bounds on  $\ret(s, d,v)-s$.

\begin{lemma}
\label{lem:np_helper}
 For $s \ge 0$ and a delivery point $d=[0,y]$ with $v/4 \le y \le v/2$ we have 
\[\frac{2y}{v} < \ret(s,d,v) - s < \frac{2y}{v} + \frac{1+4s^2 + s}{v^2-1}.\]
\end{lemma}

\begin{proof}
Let $\Delta s := \ret(s,d,v) - s$. Then the distance travelled by the drone is $\sqrt{s^2 + y^2} + \sqrt{(s+\Delta s)^2 + y^2}$. Since the drone travels at speed $v$, the time taken by the drone is then 
\[\frac{\sqrt{s^2 + 2y^2} + \sqrt{(s+\Delta s)^2}}{v}.\] 
During the delivery, the truck travels distance $\Delta s$ at speed $1$ taking the time $\Delta s$. Equating the two times we get:
\[\frac{\sqrt{s^2 + y^2} + \sqrt{(s+\Delta s)^2 + y^2}}{v} = \Delta s.\]
 Solving for $\Delta s$, we obtain:
 \[\Delta s = \frac{2(v \sqrt{s^2 + y^2} + s)}{v^2-1}.\]
 From this expression we immediately obtain the lower bound on $\Delta s$ using $s \ge 0$:
 \[\Delta s \ge \frac{2vy}{v^2-1} \ge \frac{2y}{v}.\]
 Next observe that $\sqrt{s^2 + y^2} \le y + \frac{s^2}{2y}.$ Plugging this inequality into the expression for $\Delta s$ we obtain:
 \begin{align*}
     \Delta s &\le \frac{2(v(y + s^2/2y)+s)}{v^2-1} = \frac{2vy + \frac{v}{y} s^2 + 2s}{v^2-1} = \frac{2\left(1-\frac{1}{v^2}\right) vy + 2\frac{1}{v^2} vy+ \frac{v}{y} s^2 + 2s}{v^2-1}\\
     &= \frac{2y}{v} + \frac{2\frac{y}{v} + \frac{v}{y} s^2 + s}{v^2-1} \le \frac{2y}{v} + \frac{1 + 4s^2 + s}{v^2-1},
 \end{align*}
 where in the last inequality we used the fact that $v/4 \le y \le v/2$.
 \end{proof}

\section{Strong NP-hardness}
\label{sec:np}

In this section we prove that the following decision problem is strongly NP-hard:
\begin{problem}[Schedule Length problem]
Given an instance ${I}$ of the truck-drone problem, and an integer $p$, is there a schedule $\mathcal{S}_{I}$ of length $p$ (that is, $\mathcal{S}_{I}$ makes $p$ deliveries)?
\end{problem}
We show below that there is a polynomial reduction from the well known 
{\em $3$-Partition} problem~\cite{NPbook} to the  Schedule Length problem. Recall that in the $3$-Partition problem we are given a multi-set of integers $Y = \{y_1 \le y_2 \le \cdots \le y_n\}$, where $n = 3k$. Let $T = \sum_{i=1}^n y_i/k$. 
The $3$-Partition problem asks if there is a partition of $Y$ into $k$ triples, such that the sum of elements in each triple is equal to $T$. The $3$-Partition problem is strongly NP-hard ~\cite{NPbook}.

\begin{theorem}
\label{thm:3part}
The Schedule Length problem  is strongly NP-hard.
\end{theorem}

\begin{proof}
We prove the theorem by exhibiting a reduction from a $3$-Partition instance $Y=\{y_1, y_2, \ldots, y_n \}$ to an instance $I$ of the Schedule Length problem. We use the notation for the $3$-Partition instance $Y$ introduced immediately prior to the statement of the theorem. We assume that $n$ is sufficiently large;  the values in $Y$ are bounded from above by a polynomial in $n$, so that $n^c < T \le n^{c+1}$ for a sufficiently large constant $c$.

We now define the corresponding instance $I$ of the Schedule Length problem as follows. The speed of the drone is set to $v = T$ and the flying range of the drone is set to $R=4T$. Then the minor radius of the ellipse corresponding to the speed and range of the drone is $m = 2\sqrt{T^2-1}$.  

For this proof, we depart from our convention of the truck starting at $[0,0]$ and instead specify the starting position of the truck as $[2, 0]$ (this does not affect the complexity of the problem, but makes some of the formulas nicer). 
The set of delivery points $D$ is partitioned into three subsets called $A, B$ and $C$, that are defined below:

\noindent $A = \{[0, y_1], [0, y_2], \ldots, [0, y_n]\}$ is a set of delivery points located on the $y$-axis and correspond to the inputs to the $3$-Partition problem.

\noindent $B = \{[6 + \epsilon(n), m], [2(6+\epsilon(n)), m], \ldots, [(k-1)(6+\epsilon(n)), m]\}$ and 

\noindent $C =\{[k(6+\epsilon(n)),m],[k(6+\epsilon(n))+4,m],\ldots,[k(6+\epsilon(n))+4T,m]\}$ 

\noindent are sets of delivery points that are located at distance $m$ from the $x$-axis and 
 $\epsilon(n) \in(0,1)$ is a function of $n$ to be specified later.
 
 Observe that each delivery point in $B \cup C$ can be reached by the drone from exactly one location  on the $x$-axis, and the drone must fly its full range $R=4T$ to make the delivery and return to the truck, and therefore, each such delivery takes time $R/v=4$. 
See Figure~\ref{fig:3part} for an illustration of the instance $I$ produced by the reduction, as well as the unique feasible drone trajectories for delivery points in $B$ and $C$.

In total there are $n + (k-1) + T+1$ delivery points, and we set $p = n + (k-1) + T+1$ in the Schedule Length problem instance. In other words, this instance asks whether there is a schedule that delivers to {\em all} the delivery points. Observe that the number of points and their coordinates are all bounded by a polynomial in $n$, so the reduction runs in polynomial time. 

\begin{figure}[ht]
\centering
\includegraphics[scale=0.8]{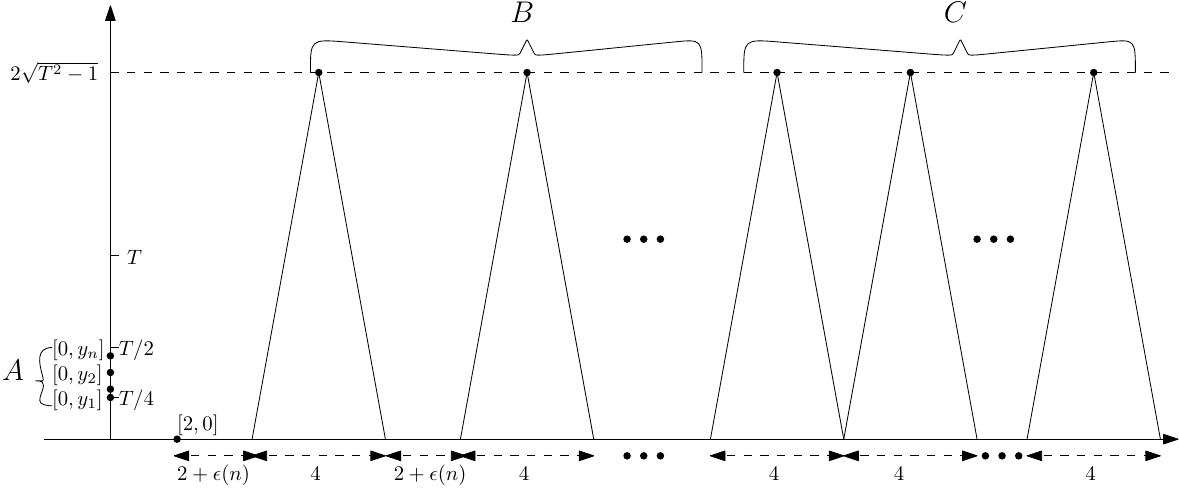}
\caption{ Illustration of  the Schedule Length  instance $I$ output by the reduction from the $3$-partition problem. The unique feasible drone trajectories for delivery points in $B$ and $C$ are also shown.}
  \label{fig:3part}
\end{figure}

We claim that the instance $Y$ to the $3$-Partition problem is a yes-instance if and only if $I$ is a yes-instance to the Schedule Length problem. 
It is clear that since the  flying range of the drone equals $4T$, no deliveries to points in $A$ can be scheduled after the deliveries to points in $C$ are made. Thus a valid schedule delivering to all the points must schedule deliveries to $A$ in the intervals between deliveries to points in $B$. There are $k$ such intervals, and each interval is of length $2 + \epsilon(n)$. We claim that at most three points $[0, y_{i_1}], [0,y_{i_2}], [0,y_{i_3}]$ can be scheduled within such an interval and if only if $y_{i_1} + y_{i_2} + y_{i_3} \le T$. Establishing this claim would finish the proof of the theorem.

Assume we have three integers $y_{i_1}, y_{i_2}, y_{i_3}$ such that 
 $y_{i_1}+ y_{i_2}+ y_{i_3} \leq T$ and the truck with the drone on it
is at position   $[i(6+\epsilon(n)+2,m]$ for $0\leq i \leq k-1$. By the upper bound on the delivery time  in Lemma~\ref{lem:np_helper} and observing that $i < k = n/3$, the total time for the three consecutive deliveries started at  $[i(6+\epsilon(n)+2),m]$ is
at most 
\begin{align}
\frac{2(y_{i_1} +y_{i_2}+ y_{i_3})}{T} +3 \frac{1+ 4 (k(6+\epsilon(n)))^2 + (k(6+\epsilon(n)))}{T^2-1} \leq 2 + \frac{2n^2(6+\epsilon(n))^2}{T^2-1} = 2 + O(n^2/T^2).
\end{align}
Thus, the deliveries to $[0, y_{i_1}], [0,y_{i_2}], [0,y_{i_3}]$  can be completed before the delivery to   $[(i+1)(6+\epsilon(n)+2,m]$ is scheduled, provided that $O(n^2/T^2) = O(n^{-2c+2}) \le \epsilon(n)$. 

Assume we have three integers $y_{i_1}, y_{i_2}, y_{i_3}$ such that 
 $y_{i_1}+ y_{i_2}+ y_{i_3} > T$ and the truck with the drone on it
is at position   $[i(6+e(n)+2,m]$ for $0\leq i \leq k-1$. By the lower bound on the delivery time  in Lemma~\ref{lem:np_helper}, the total time for the three consecutive deliveries started at  $[i(6+\epsilon(n)+2,m]$ is
at least $2(y_{i_1} +y_{i_2}+ y_{i_3})/T > 2 + 1/T$ and they cannot be completed before the delivery to   $[(i+1)(6+\epsilon(n)+2,m]$ is scheduled, provided that the term $1/T \ge n^{-c-1}$ exceeds $\epsilon(n)$. 

It is left to notice that because we can take $n$ and $c$ sufficiently large, we can find $\epsilon(n)$ satisfying:
\[ O\left(\frac{1}{n^{2c-2}}\right) < \epsilon(n) < \frac{1}{n^{c+1}}.\]
For example, one could take $c = 4$ and $\epsilon(n) = 1/n^6$. This completes the proof of strong NP-hardness. 
\end{proof}

\section{A Greedy Approximation Algorithm}
\label{greedy-alg}
\noindent
In this section we describe a greedy scheduling algorithm for the truck-drone problem.  Our algorithm, which we call $\mathcal{A}_g$ , assigns deliveries to the drone as the truck moves from  left to right starting from the initial position of the truck at $[0, 0]$.  
When the truck with the drone is at position $[s,0]$, our greedy algorithm schedules a delivery to point $d$ which, from among all feasible delivery points, minimizes the round-trip flight time from $[s,0]$, i.e., which gives the {\em earliest possible return} for the drone  to the truck. 
Notice that the delivery point which  minimizes the round-trip flight time from $[s,0]$
is not necessarily the delivery point that is at the shortest distance from $[s,0]$. For example, in Figure \ref{twoappex}, the point $d_1$ is closer than $d_2$ to  $[s,0]$.  Thus one needs to use the function defined by 
Formula \ref{return} to calculate  which delivery point requires the  shortest time to return to the truck.  We then update $s$ to be this shortest return time. If there are no feasible delivery points, then $s$ i set to the earliest time any of the remaining points can be reached after the current time. 

\begin{figure}[ht]
\centering
\includegraphics[scale=0.5]{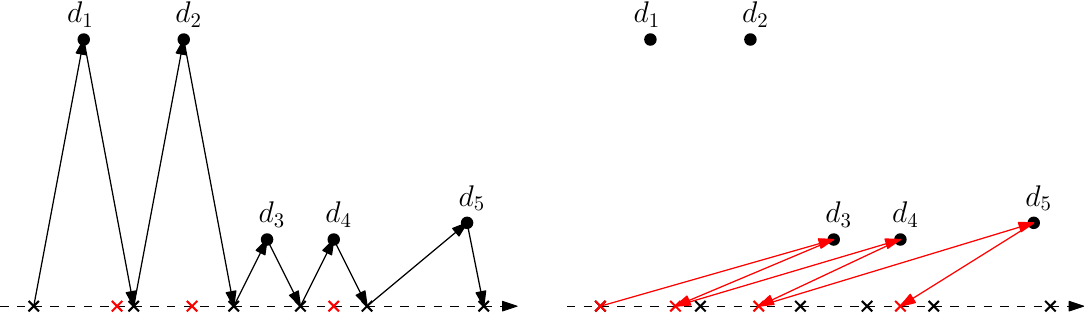}
\caption{ The black arrows,  red arrows show the travel of the drone according to an optimal, greedy schedule, respectively,   for an instance 
${I}=(0,4,8,D)$ with $D$ containing five delivery points $\{d_1,d_2,\ldots,d_5\}$. The black crosses and red crosses on the $x$-axis indicate the return points of $OPT$ and $\mathcal{A}_g$, respectively.}
  \label{fig:opt-greedy}
\end{figure}

Algorithm~\ref{alg:cap} 
gives the pseudocode for $\mathcal{A}_g$. It is straightforward to see that  Algorithm   $\mathcal{A}_g$ can be implemented in $O(n^2)$ time, since a single evaluation of $\ret$ takes constant time.  Figure \ref{fig:opt-greedy} gives an example of the trajectories of the drone according to an optimal schedule and that of the schedule calculated by  $\mathcal{A}_g$. 

\begin{algorithm}[H]
\caption{Greedy Approximation Algorithm $\mathcal{A}_g$ to Compute Feasible Delivery Schedule}\label{alg:cap}
\begin{algorithmic}[1]
\Require Instance ${I}=(v,R,D)$ where $D=\{d_1,d_2,\ldots,d_n\}$, is a list of delivery points.
\Ensure $\mathcal{S}_I$ is a feasible schedule of deliveries. 
\State $\mathcal{S}_I\gets L \gets \emptyset$
\State $s \gets 0$ \\
\Comment{For each delivery point $d_i$, calculate $es(d_i)$ and $ls(d_i)$ and insert  triple into $L$.}
\For {$i=1 \ldots n$} \If {$s \leq ls(d_i)$}
\State $x.es = es(d_i)$
  \State $x.ls = ls(d_i)$
 \State $x.d = d_i$
\State  Insert($L, x$)
\EndIf
\EndFor
\State Sort($L, key=es$)
\While{$L \neq \emptyset$}
\State $x \gets first(L)$ \\
\If {$s < x.es$}
\State $s \gets x.es$ \Comment{If no feasible delivery point, move $s$ forward.~~~~~~}
\EndIf \\
\Comment{Find feasible delivery point which minimizes the return time to truck.~~~~~}
\State $r_{min}\gets \infty$
\While{ $x \neq NIL$  and $s \geq x.es$}
\State $r \gets \ret(s, v, x.d_i)$
\If {$r < r_{min}$}
\State $r_{min} \gets r$
\State $save \gets x$
\EndIf
\State $x \gets next(L)$
\EndWhile \\
\Comment{Insert next delivery point into schedule, update $s$ and list $L$~~~~~~~~~~~~~~~~~~~}
\State Insert ($\mathcal{S}_I, (save.d, s)$)
\State $s \gets r_{min}$
\For {$x \in L$}
\If {$x.lr < s$}
\State Delete ($L, x$)
\EndIf
\EndFor
\EndWhile
\end{algorithmic}
\end{algorithm}

In the next theorem we compare the size of the schedule calculated by  Algorithm  $\mathcal A_g$ with respect to an optimal algorithm.

\begin{theorem}
\label{th:comp}
Given an instance $I=( v, R, D)$    of the truck-drone delivery problem,  let $\mathcal{S}_g$ be the schedule produced by the algorithm $\mathcal{A}_g$ and let $\mathcal{S}_{OPT}$ be an optimal schedule. Then $$|\mathcal{S}_{OPT}| \leq 2 |\mathcal{S}_g|$$

\end{theorem}
\begin{proof}

Let $D=\{d_1,d_2,\ldots,d_n\}$ and let
 $$\mathcal{S}_{OPT}=((d_{i_1},s_1),(d_{i_2},s_2),\ldots, (d_{i_p},s_p)), \text{ and }$$ 
  $$\mathcal{S}_{g}=((d_{j_1},s'_1),(d_{j_2},s'_2),\ldots, (d_{j_q},s'_q)), \text{where } q \le p \le n.$$ 

  We give a function $\mathcal{F}$ that maps delivery points in $\mathcal{S}_{OPT}$ to points in $\mathcal{S}_g$. For every $k$, with $1 \leq k \leq p$, define 
  $r_k$ to be the return time of the drone for the $k^{th}$ delivery in $S_{OPT}$ and similarly 
  for every $k$, with $1 \leq k \leq q$, define 
  $r'_k$ to be the return time of the drone for the $k^{th}$ delivery in   $\mathcal{S}_g$. Define $Q_k$ to be the set of delivery points in $\mathcal{S}_g$ whose return to the truck in  the greedy schedule occurs  {\em during}  the flight time of the drone to deliver the $k^{th}$ item in $\mathcal{S}_{OPT}$. That is, $$Q_k = \{ d_{j_\ell} \:: r'_\ell \in (s_k, r_k] \}$$

  If $Q_k \neq \emptyset$, define $last(Q_k)$ to be the element of $Q_k$ with the latest return according to the greedy schedule. 

  Now define $P_k$ to be the set of delivery points in $\mathcal{S}_g$ whose start time in the greedy schedule is before the start time of the $k^{th}$ delivery in the optimal schedule, but whose return to the truck in the greedy schedule occurs between the return from the $k^{th}$  delivery in the optimal schedule  and the start of the $(k+1)^{st}$ delivery. 
  That is:
  
  $$P_k = \{ d_{j_\ell} \::  s'_\ell \leq s_k \mbox{ and } r_k < r'_\ell \leq s_{k+1} \}$$

  If $P_k \neq \emptyset$, note that it can have only one element, denote it as $p_k$. 

  We are now ready to define the function $\mathcal{F}$. For all $k \in \{1, \ldots, p \}$

\begin{subnumcases}{\mathcal{F}(d_{i_k})=}
       last(Q_k)    & \text{if } $Q_k \neq \emptyset $ \label{func1} \\
        p_k & \text{if } $Q_k = \emptyset$ and $P_k \neq \emptyset$ \label{func2} \\
       d_{i_k} & otherwise    \label{func3}         
      \end{subnumcases}

We give an example to illustrate function $\mathcal{F}$ using 
an  instance shown in Figure \ref{fig:opt-greedy}. In that case  the optimal schedule makes 5 deliveries in order to 
 $(d_1,d_2,d_3,d_4,d_5)$ and greedy schedule  contains
3 deliveries $(d_3,d_4, d_5)$ listed in order, omitting the starting times.  For this case the function $\mathcal{F}$ is as follows:

$\;\;\mathcal{F}(d_1)=d_3$, $\mathcal{F}(d_2)=d_4$,
$\mathcal{F}(d_3)=d_3$, $\mathcal{F}(d_5)=d_5$, and  $\mathcal{F}(d_5)=d_5$. 

First we prove that  Clauses \ref{func1}, \ref{func2}, and \ref{func3} 
 define  a valid function on $L'$, that is, every delivery point in the optimal schedule is mapped to a delivery point in the greedy schedule. 
Since $Q_k$ and $P_k$ only contain delivery points in the greedy schedule, the only case to consider is that $Q_k=P_k = \emptyset$ and $\mathcal{F}(d_{i_k
}) = d_{i_k}$ and $d_{i_k}$ is not part of the schedule $\mathcal{S}_g$ of the greedy algorithm $\mathcal{A}_g$. 



Let $\ell$ be the largest integer such that $s'_\ell \leq s_k$. By assumption $d_{j_\ell} \neq d_{i_k}$. Since
$Q_k = P_k = \emptyset$, either $r'_\ell > s_{k+1}$ or $r'_\ell \leq s_k$.
If $r'_\ell \leq s_k$ (see Figure~\ref{fig:caseab}(a)), consider  the $(\ell+1)^{st}$ delivery by the greedy algorithm. We know that $s'_{\ell+1} > s_k$ and since  $Q_k = \emptyset$, it must be that $r'_{\ell+1} > r_k $.  Thus for its $(\ell+1)^{st}$ delivery, the greedy heuristic should have chosen to deliver to $d_{i_k}$ rather than to $d_{j_{\ell+1}}$, a contradiction. 

Therefore it must be that $r'_\ell > s_{k+1}$. But then, using Lemma~\ref{lem:valid-trajectories}, there is a valid trajectory for the drone flying to $d_{i_k}$ starting at  $s'_\ell$  with an  {\em earlier} return time that is at most $ r_k \leq s_{k+1} < r'_\ell$ (see Figure~\ref{fig:caseab}(b)). Thus for its $\ell^{th}$ delivery, the greedy heuristic should have chosen to deliver to $d_{i_k}$ rather than to $d_{j_\ell}$, a contradiction. Thus $d_{i_k}$ must be part of the greedy schedule, and  $\mathcal{F}$ is a valid function mapping the delivery points in $\mathcal{S}_{OPT}$ to the delivery points in $\mathcal{S}_g$.

\begin{figure}[ht]
 \centering
\includegraphics[scale=0.5]{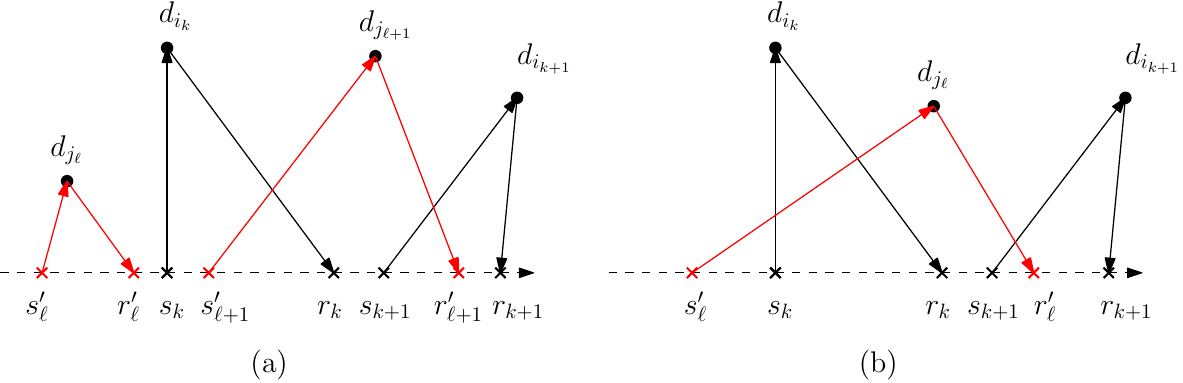}
\caption{This figure illustrates the two cases in the proof that the function $\mathcal{F}$ in Theorem~\ref{th:comp} is well defined. Red lines show deliveries in a presumed greedy schedule, and black lines indicate deliveries in a presumed optimal schedule.The dashed line represents the $x$ axis, with distances from the origin marked.}
  \label{fig:caseab}
  \end{figure}


Finally, we claim that $\mathcal{F}$ maps at most two delivery points in $\mathcal{S}_{OPT}$ to one delivery point in $\mathcal{S}_g$. First, since the half-closed intervals $(s'_1,r'_1],(s'_2,r'_2],\ldots,(s'_{k'},r'_{k'}]$
are all disjoint,  and  the half-closed intervals   $(s_1,r_1],(s_2,r_2],\ldots,(s_{k'},r_k]$ are also all disjoint, 
and  any return point $[r',0]$  can satisfy at most one of  Clauses \ref{func1} and \ref{func2},
it follows that distinct  elements in $\mathcal{S}_{OPT}$ are mapped to distinct elements  of $\mathcal{S}_g$ by those two clauses. Second, clearly distinct  elements in $\mathcal{S}_{OPT}$ are mapped to distinct elements  of $\mathcal{S}_g$ by Clause~\ref{func3}. Therefore, the only kind of "collision" that can occur is that $\mathcal{F}(d_{i_k})$ is mapped to $d_{j_\ell}$ by Clause ~\ref{func1} or Clause \ref{func2} and $\mathcal{F}(d_{i_q})$ is mapped to $d_{j_\ell}$ by Clause~\ref{func3}. This proves our claim that $\mathcal{F}$ maps at most two delivery points in $\mathcal{S}_{OPT}$ to one delivery point in $\mathcal{S}_g$. 


We conclude that  the schedule $\mathcal{S}_g$ created by  
 $\mathcal{A}_g$  contains  at least  $\lceil p/2\rceil$ elements, as desired. 
 That is, $\mathcal{A}_g$ is a 2-approximation algorithm.
\end{proof}

The approximation ratio of 2 is tight. To see this, consider the instance given in Figure \ref{twoappex}. For this  instance  the schedule computed by the greedy algorithm contains exactly one half of the delivery points, while an optimal schedule makes  deliveries to all  points. Thus, the approximation factor of $2$ in Theorem \ref{th:comp} cannot be improved.   
\begin{figure}[ht]
\centering
\includegraphics[width=0.95\textwidth]{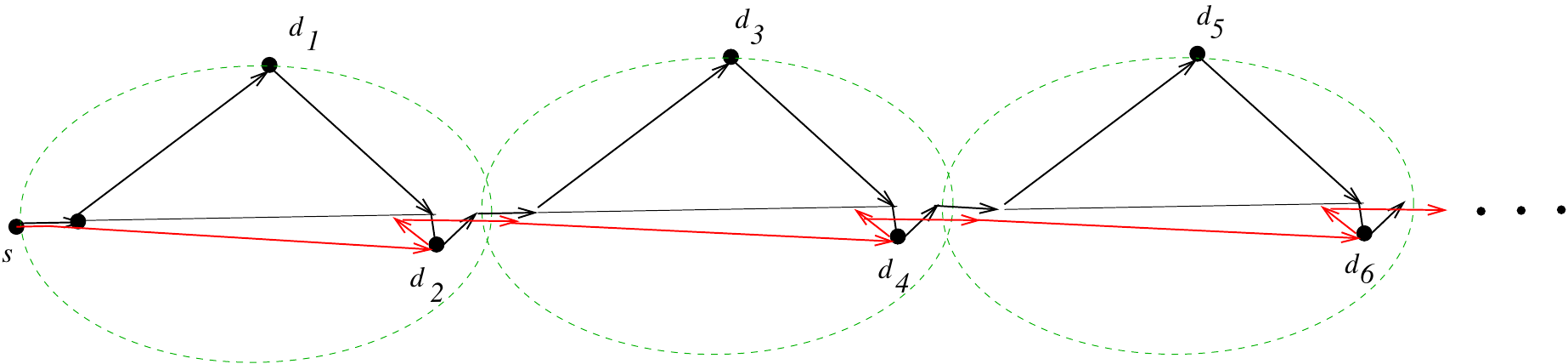}
\caption{Approximation factor of $2$ is sharp: in this instance, $d_1,d_3,d_5,\ldots$ are located at the maximal reach of the drone, and reachable only from the left focus  of the corresponding ellipse (dashed green). An optimal schedule, shown in black, contains  all points in order $d_1,d_2,d_3,d_4,\ldots$. 
The greedy algorithm,  can immediately schedule a delivery to $d_2$, but not to $d_1$. After scheduling a delivery to $d_2$ a delivery
to $d_1$ is not feasible any more, and this scheduling, shown in red, is repeated, resulting in the  schedule   $d_2,d_4,d_6,\ldots$.} 
  \label{twoappex}
\end{figure}

 \section{Optimal algorithm for a restricted set of inputs}
\label{sec:optimal}

As seen in the proof of strong NP-hardness in Section~\ref{sec:np}, having many delivery points
with the same $x$-coordinate creates a decision problem: should a delivery to a  point $[0,y]$ be scheduled  prior to or after the truck reaches $[0,0]$. These decisions make the truck-drone problem  NP-hard. In this section, we specify a family of instances called \IN instances in which the delivery points do not  have the same  or ``nearly'' the same $x$-coordinates, where ``nearly'' depends on the difference in their $y$-coordinates.  In particular, the greater the difference in the $y$-coordinates of the points, the greater is the difference in their $x$-coordinates in \IN instances. We show that there is $O(n^3)$ algorithm to compute an optimal schedule for \IN instances.

\begin{definition}\label{def:schedule}
 Let  $I=(v, R, D)$  be an instance of the truck-drone 
delivery problem where $D=\{d_1,d_2,\ldots,d_n\}$.  We say $I$ is a \IN instance if:
\begin{itemize}
    \item for every $i, j \in \{1, \ldots, n \}$, with $i \neq j$, the delivery point $d_j$ is not contained in the triangle $[es(d_i),0],d_i,[lr(d_i),0]$, and 
\item the set of closed intervals 
$\{[es(d_1), ls(d_1)], [es(d_2), ls(d_2)],
\ldots, [es(d_n), ls(d_n)]\}$
form a proper interval graph \cite{golumbic}, i.e., no interval in the set is a subset of another interval in the set.  
\end{itemize}
\end{definition}

Figure \ref{fig:prime} shows an example of a   \IN  instance. The definition of a \IN instance implies that the delivery points  have pairwise 
different $x$-coordinates and clearly, not many of them can reside in a narrow vertical band. 

The  lemma below implies that for a \IN instance with  $D=\{d_1,d_2,\ldots,d_n\}$, 
the intervals  $[es(d_1), ls(d_1)], [es(d_2), ls(d_2)],
\ldots, [es(d_n), ls(d_n)]$ are ordered by the $x$-coordinates of the corresponding points in $D$.
\begin{lemma}
\label{lemma:order}
Let $I=( v, R, D)$   be a \IN  instance of the truck-drone delivery problem with $D=\{d_1,d_2,\ldots,d_n\}$. Let  $d_i=[x_i,y_i]$ and $d_j=[x_j,y_j]$ be two points of $D$  with $x_i<x_j$. 
Then either $ ls(d_i)< es(d_j)$, or 
$es(d_i)<es(d_j)\leq ls(d_i)<ls(d_j)$
\end{lemma}
\begin{proof}

If $y_j < y_i$ then $ls(d_j) >ls(d_i)$. Since interval $[es(d_j), ls(d_j)]$
cannot contain $[es(d_i), ls(d_i)]$, either  $ ls(d_i)< es(d_j)$, or 
$es(d_i)<es(d_j)\leq ls(d_i)<ls(d_j)$.

 If $y_j > y_i$ then $es(d_j) >es(d_i)$. Since interval $[es(d_i), ls(d_i)]$
cannot contain $[es(d_j), ls(d_j)]$, either  $ ls(d_i)< es(d_j)$, or 
$es(d_i)<es(d_j)\leq ls(d_i)<ls(d_j)$.
\end{proof}
Given an instance $I=(v, R, D)$, we can verify if $I$ is a \IN instance  in $O(n^2)$ time by checking each pair of intervals for non-containment,  and each triangle for the non-inclusion of other points of $D$.

\begin{figure}[ht]
\centering
\includegraphics[width=0.9\textwidth]{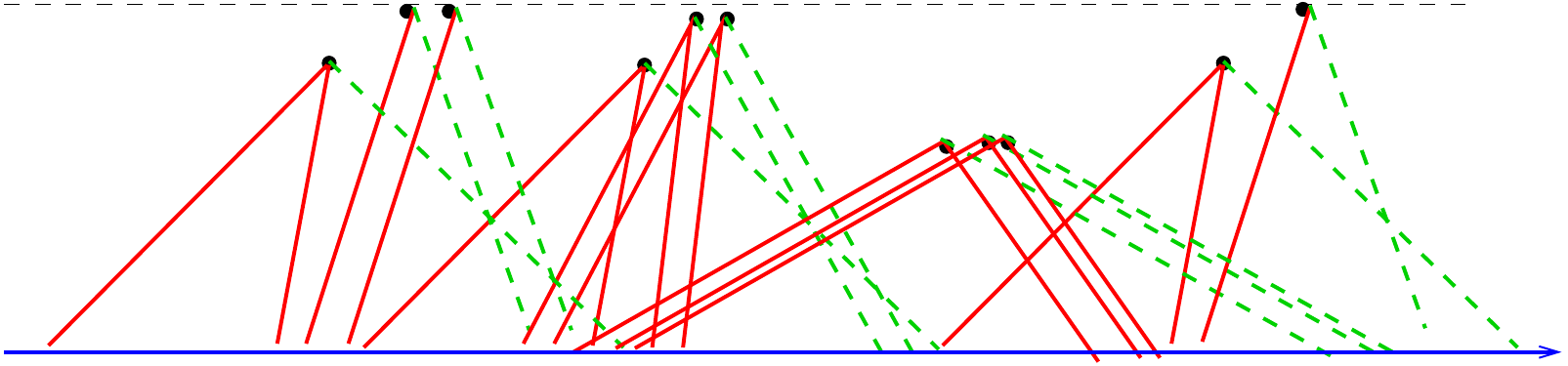}
\caption{An example of a \IN instance $I=(v, R, D)$ position for $v=3$, and $r=12$. For each delivery point the red segment points to the corresponding $ed$ and $ld$ points on the line, and the green segment points to the corresponding $la$. The three topmost points are at the limit of the reach of the drone. }
  \label{fig:prime}
\end{figure}

The following lemma is used to show that for {\IN} instances we can restrict our attention to schedules in which the  subsequent deliveries 
are ordered by the $x$ -coordinates of delivery points, and in which the trajectories of the drone are non-crossing.   See Figure~\ref{fig:traj} for an illustration. 

\begin{figure}[ht]
\centering
\includegraphics[scale=0.7]{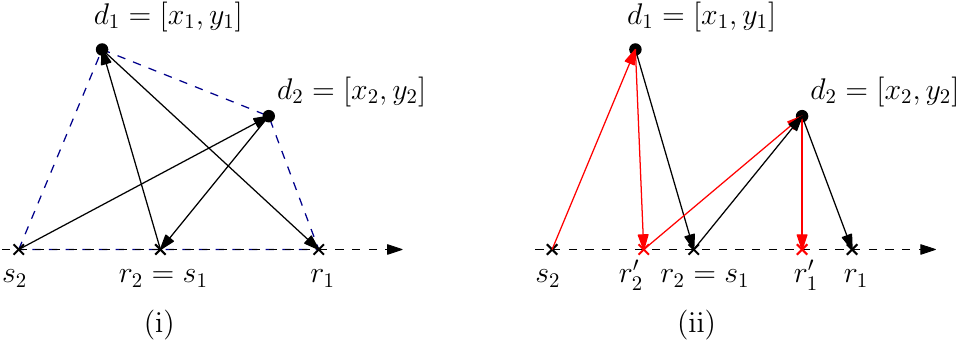}
\caption{(i) The crossing trajectories to $d_1$ and $d_2$ are shown in black.
(ii) The non-crossing trajectories, shorter in total, are in red.}
  \label{fig:traj}
\end{figure}

\begin{lemma}
\label{lemma:noncros}
Let $I=(v, R, D)$   be a \IN instance of the truck-drone problem. Assume that there is a feasible schedule for this instance 
in which a delivery to, say  $d_2=[x_2,y_2]$ immediately precedes that to  $d_1=[x_1,y_1]$, with $x_1<x_2$. Then
\begin{enumerate}
\item The trajectories of the drone to $d_1$ and $d_2$ must cross.
\item By swapping the order of deliveries to $d_1$ and $d_2$ the total time of the two deliveries cannot increase, and thus swapping the two deliveries
maintains the feasibility of the schedule, i.e., crossings of  two consecutive trajectories can be avoided.
\end{enumerate}
\end{lemma}
\begin{proof}

\begin{figure}[ht]
\centering
\includegraphics[scale=0.7]{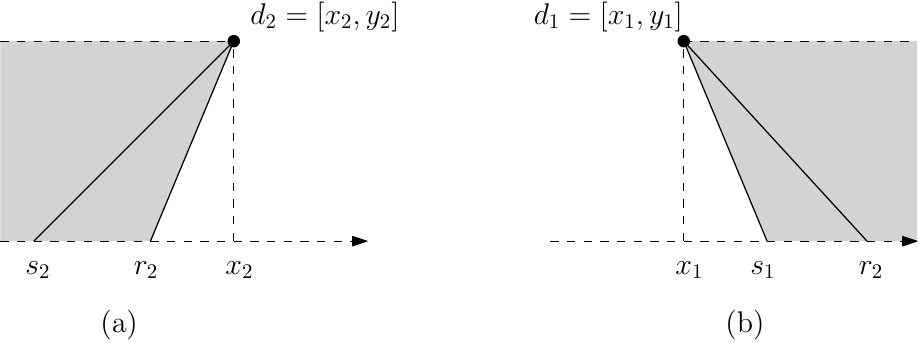}
\caption{Illustration for the proof of (1) in Lemma~\ref{lemma:noncros}. $d_1=[x_1,y_1]$ and $d_2 = [x_2,y_2]$ with $x_1 < x_2$ and delivery to $d_2$ occurring before the delivery to $d_1$. Figure (a) illustrates the case of $y_1 \le y_2$ and the shaded region demonstrates locations of $d_1$ which result in a crossing trajectory. Figure (b) illustrates the case of $y_1 \ge y_2$ and the shaded region demonstrates locations of $d_2$ which result in a crossing trajectory.}
\label{fig:noncros-case-1}
\end{figure}

To see (1), let $s_i, r_i$ denote the start and return times to delivery point $d_i$ for $i \in \{1,2\}$. Assume for contradiction that delivery trajectories do not cross. If $y_1 \le y_2$, then $d_1$ lies inside the triangle $[r_2,0], d_2, [x_2,0]$. This triangle is clearly contained in $[es(d_2),0], d_2, [lr(d_2),0]$ contradicting $D$ being proper. If $y_1 \ge y_2$ then $d_2$ lies inside the triangle $[x_1,0], d_1, [s_1,0]$, which is contained inside $[es(d_1), 0], d_1, [lr(d_1),0]$. This also contradicts $D$ being proper. See Figure~\ref{fig:noncros-case-1}.

Next, we show (2). By Observation \ref{obs:2} we can assume that the delivery to   $d_1$ starts immediately at time $r_2$, i.e., $s_1=r_2$ and terminates at time  $r_1$.
Clearly, in this case $es(d_2) <ls(d_1)$ and thus, by Lemma \ref{lemma:order},
$es(d_1)<es(d_2)\leq s_2 <r_2\leq  ls(d_1)<ls(d_2)$. Thus, a delivery to $d_1$ can be started at time  $s_2$, and a delivery to $d_2$ can be started at time  $r_2$ or later. 
It remains to show that the reversal in the delivery order 
can terminate latest at time $r_1$.

Suppose first that delivery to $d_1$, when started at time  $s_2$ takes at most as much  time as a delivery to $d_2$ at time $s_2$, see Figure \ref{fig:traj} (ii). In the paragraph below, we use $s_i$ to denote the point $[s_i, 0]$ and similarly $r_i$ to denote the point $[r_i, 0]$.
Consider the quadrilateral  $s_2$,$d_1$,$d_2$,$r_1$ shown in blue. Since our instance is a \IN instance, the triangle  $s_2$,$d_1$,$r_1$
doesn't contain $d_2$ and thus this quadrilateral is convex.  
 By the triangular inequality the sum  $|s_2,d_1| +|d_2,r_1|$  of the lengths of two opposite sides  of the quadrilateral  is strictly less than the 
sum of the length of its diagonals $|d_1,r_1|+|s_2,d_2|$.  
Therefore,
$$|s_2,d_1| +|d_1,r_2| + |r_2,d_2|+ |d_2,r_1| 
<|s_2,d_2| +|d_2,r_2|+|r_2,d_1|+|d_1,r_1|$$
and the  path $s_2,d_1,r_2,d_2,r_1$  is shorter than the trajectory $s_2,d_2,r_2,d_1,r_1$.
However, the path $s_2,d_1,r_2,d_2,r_1$ is {\em not necessarily a valid drone trajectory} if the delivery to $d_1$ from $s_2$ takes less time than the  delivery to $d_2$ from $s_2$. Then, 
when  delivering to $d_1$ first, 
the drone returns to the truck at point  $r_2'$ located strictly between $s_2$ and $r_2$.  But then the path $s_2,d_1,r_2',d_2,r_1$ is even shorter than path $s_2,d_1,r_2,d_2,r_1$. Thus, when starting 
the delivery to $d_2$ at $r_2'$, the drone returns to the truck at a point $r_1'$ to the left of $r_1$, which improves the total delivery time to $d_1$ and $d_2$.  

\begin{figure}[ht]
\centering
\includegraphics[scale=0.7]{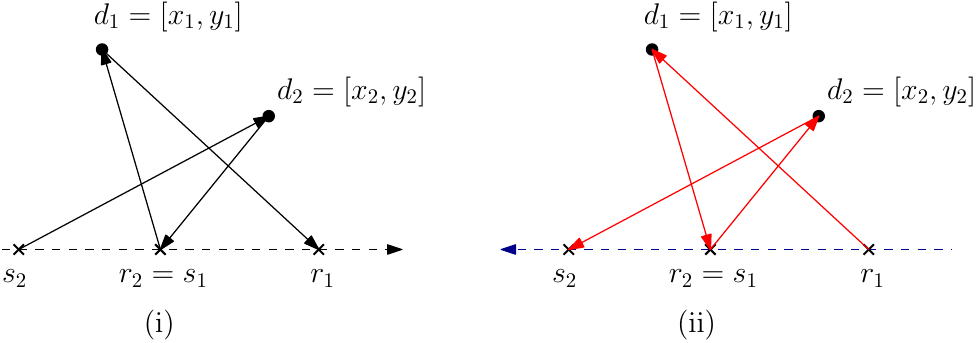}
\caption{ Reversing the directions of deliveries and the movement of the truck in (ii) converts a configuration of the second case to the first case}
\label{fig:noncros2}
\end{figure}

Now suppose instead that delivery to $d_1$, when started from $s_2$, takes more time then  the delivery to $d_2$ from $s_2$, as for example on Figure \ref{fig:noncros2} (i). By the shape of the function $\ret(s, d,v))$, see Observation \ref{obs:profile}, and since 
$es(d_1)<es(d_2)$ and $ls(d_1) <ls(d_2)$,  a delivery to $d_1$, when started from $s_1$ also takes more time than  the delivery to $d_2$ from $s_1$.
Consider the configuration on Figure \ref{fig:noncros2}(ii)
in which we reverse the movement of the drone and of the truck. Then we reduced
this to the previous case and a delivery from $s_1$ first to  $d_2$ and then to $d_1$ is shorter, and by reversing this once more we obtain that the delivery from $s_2$ first to $d_1$ and then to $d_2$ is shorter.   
\end{proof}

A \IN instance is guaranteed to have an optimal schedule with non-crossing trajectories. However not all optimal schedules give non-crossing trajectories. Indeed 
there are non-proper instances  where crossing of trajectories is required in any optimal schedule as demonstrated in Figure \ref{fig:crossingneeded}.

\label{sec:figures}
\begin{figure}[ht]
\centering
\includegraphics[width=0.5\textwidth]{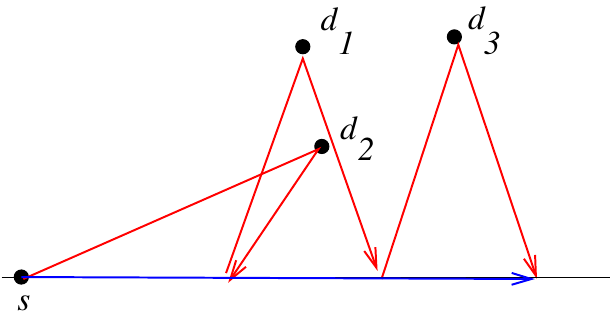}
\caption{ An instance of the problem where any optimal schedule
must contain crossing trajectories.  
Points $d_1$ and $d_3$ are at the maximum reach of the drone.   
When scheduling a delivery first to $d_1$ then 
 a delivery is possible either to $d_2$ or to $d_3$, but not to both. }   
\label{fig:crossingneeded}
\end{figure}

\begin{definition} Let  $I=(v, R, D)$  be an instance of the truck-drone 
delivery problem where $D=\{d_1,d_2,\ldots,d_n\}$. 
We  call  schedule $$\mathcal{S}_I=((d_{i_1},s_1),(d_{i_2},s_2),\ldots, (d_{i_m},s_m), m\leq n$$ {\em \MON} if the $x$-coordinate of $d_{i_j}$ is strictly 
less than the  $x$-coordinate of $d_{i_{j+1}}$ for every $1\leq j\leq m-1 $,  
\end{definition}
In the next theorem we show that there always exists a monotone schedule with the optimal substructure property 
for \IN instances.

\begin{theorem}
\label{th:monotone}
Let  $I=(v, R, D)$  be a \IN instance of the truck-drone 
delivery problem with $D=\{d_1,d_2,\ldots,d_n\}$. 
Assume that the points in $D$ are listed according to increasing $x$-coordinate. Then there is an optimal schedule 
 $\mathcal{S}_I=((d_{i_1},s_1),(d_{i_2}.s_2),\ldots, (d_{i_m},s_m)$, $m\leq n$
 for this instance with the following properties:\
\begin{enumerate}
\item  $\mathcal{S}_I$ is monotone.
\item For every $j\leq m$, the initial part $((d_{i_1},s_1),(d_{i_2},s_2),\ldots, (d_{i_j},s_j))$ of  $\mathcal{S}_I$ 
 minimizes the delivery completion time for any    
subset of $\{d_1,d_2, \ldots,d_{i_j}\}$ of size $j$.
\end{enumerate} 
\end{theorem}
\begin{proof}
Assume $I=(v, R, D)$  is a given \IN instance of the truck-drone 
delivery problem  and let $\mathcal{S}_I$ be an 
optimal schedule for it.
By a repeated application of Lemma \ref{lemma:noncros}
we can  swap any two consecutive deliveries that don't respect the 
order of $x$-coordinates of points, as in the bubble sort, while maintaining 
the schedule optimal. This eventually produces a
monotone schedule of the same (optimal) length, proving (1). 

To show (2), assume that for some  $j\leq m$, there is a subset of $j$ points 
$\{d_{i_1'},d_{i_2'},\ldots, d_{i_j}'\}$ of  the set $\{d_1,d_2, \ldots,d_{i_j}\}$
for which there is a schedule  $((d_{i_1}',s_1'),$ $(d_{i_2}',s_2'),\ldots, (d_{i_j}',s_j'))$
with $s_j'<s_j$ and which minimizes the delivery completion time for any subset of  $\{d_1,d_2, \ldots,d_{i_j}\}$ of size $j$.
Then by concatenating  $((d_{i_1}',s_1'),(d_{i_2}',s_2'),\ldots, (d_{i_j}',s_j'))$ with  $((d_{i_{j+1}},s_{j+1}),\ldots, (d_{i_m},s_m)$,
we get a valid schedule.
In this manner, repeating the process starting with  $j=i_m$ 
and decreasing appropriately the value of $j$  we can get a schedule for $I$ that is optimal, monotone, and satisfies the property 2 of the theorem. 
\end{proof}

We use Theorem~\ref{th:monotone} to describe a dynamic programming algorithm 
that finds an optimal schedule for \IN instances. 

\begin{theorem}
\label{th:monalg}
There is an $O(n^3)$ algorithm that calculates an optimal 
schedule for any \IN instance   $I=(v, R, D)$  of the truck-drone 
delivery problem. 
\end{theorem}
\begin{proof}
Assume the delivery points in $D$ are listed in the order of their $x$ coordinates. Define $T(i, j)$ to be the earliest delivery completion time for the truck and the drone to perform exactly $i$ deliveries from among $d_1, d_2, \ldots, d_j$ where $d_j$ must be included in the schedule. If such a schedule is not possible, we define $T(i,j) = \infty$. We can compute $T(i,j)$ using dynamic programming as follows. We clearly have $T(1, j) = \ret(s, d_j, v)$ for the base case of $i = 1$ (see Lemma~\ref{lem:return_time_from_starting_pos} for the definition of $\ret$) where $[s,0]$ is the starting position of the truck. For $i \ge 2$, we have $T(i, j) = \min_{j' < j} \ret(T(i-1, j'), d_j, v)$. This recursive formula immediately follows from the optimal substructure property stated in Theorem~\ref{th:monotone}: a schedule resulting in the earliest completion time of making $i$ out of the first $j$ deliveries where $d_j$ is included consists of delivering to $i-1$ out of the first $j' < j$ delivery points (with earliest completion time $T(i-1, j')$) followed by earliest delivery completion to $d_j$. Note that defining $\ret(s',v,d) = \infty$ when $s' > ls(d)$ and $T(i,j) = \infty$ when delivery is impossible correctly works with the recursive computation of $T$.

Having computed $T$, we can find the maximum number of deliveries that can completed in a valid schedule by taking the maximum $m$ such that $T(m,j) \neq \infty$ for some $j$. By recording for each $(i,j)$ pair which choice of $j'$ resulted in the table entry $T(i,j)$, we can reconstruct the schedule itself using standard backtracking techniques.

The running time is dominated by computing the table $T(i,j)$. It has $O(n^2)$ entries and each entry can be computed in time $O(n)$, since a single evaluation of $\ret$ takes constant time. The overall runtime is then $O(n^3)$.
\end{proof}

\section{Discussion}
\label{discuss}

We have shown that even in the simple case of a single drone with a single truck 
travelling in a straight line, the problem of coordinating their efforts to maximize
the number of deliveries made is hard. Our work raises a number of different questions.
We show that a greedy strategy achieves a 2-approximation. Is a better approximation
possible? In particular, is the problem APX-hard or might there be a PTAS for it?
Our implementation of the greedy strategy runs in $O(n^2)$ time. Is a better running time 
for the algorithm possible by taking advantage of the structure of the intervals created by the drone
paths? The set of proper instances includes those where the $y$-coordinate is fixed. Could this
be expanded to include points with a limited number of different $y$-coordinates? More generally,
is there a "natural" setting in which the problem becomes fixed-parameter tractable?
Finally, many variations on the problem are worth pursuing. Rather than maximizing the number
of deliveries made with a given speed or drone range, one could consider the dual problems of 
minimizing the speed or range required to complete all deliveries. Versions with multiple
drones and/or trucks, larger capacity drones, etc. are also of interest. 

\bibliographystyle{plain}
\bibliography{cite}

\begin{thebibliography}{10}

\bibitem{dr1}
A.~Cornell, B.~Kloss, and R.~Riedel.
\newblock Drones take to the sky, potentially disrupting last-mile delivery.
\newblock {\em
  https://www.mckinsey.com/industries/aerospace-and-defense/our-insights/future-air-mobility-blog/drones-take-to-the-sky-potentially-disrupting-last-mile-delivery},
  2023.

\bibitem{erlebach2022package}
Thomas Erlebach, Kelin Luo, and Frits C.~R. Spieksma.
\newblock Package delivery using drones with restricted movement areas.
\newblock In Sang~Won Bae and Heejin Park, editors, {\em 33rd International
  Symposium on Algorithms and Computation, {ISAAC} 2022, December 19-21, 2022,
  Seoul, Korea}, volume 248 of {\em LIPIcs}, pages 49:1--49:16. Schloss
  Dagstuhl - Leibniz-Zentrum f{\"{u}}r Informatik, 2022.

\bibitem{dr2}
H.~Eskandaripour and E.~Boldsaikhan.
\newblock Last-mile drone delivery: Past, present, and future.
\newblock {\em Special Issue The Applications of Drones in Logistics}, 2023.

\bibitem{freitas2023exact}
J{\'u}lia~C Freitas, Puca Huachi~V Penna, and T{\'u}lio~AM Toffolo.
\newblock Exact and heuristic approaches to truck--drone delivery problems.
\newblock {\em EURO Journal on Transportation and Logistics}, 12:100094, 2023.

\bibitem{NPbook}
M.~R. Garey and D.~S. Johnson.
\newblock {\em Computers and Intractability; A Guide to the Theory of
  NP-Completeness}.
\newblock W. H. Freeman \& Co., New York, NY, USA, 1990.

\bibitem{golumbic}
M.~C. Golumbic.
\newblock Interval graphs and related topics.
\newblock {\em Discrete Mathematics}, 55:113--121, 1985.

\bibitem{khanda2022drone}
Arindam Khanda, Federico Cor{\`o}, and Sajal~K Das.
\newblock Drone-truck cooperated delivery under time varying dynamics.
\newblock In {\em Proceedings of the 2022 Workshop on Advanced tools,
  programming languages, and PLatforms for Implementing and Evaluating
  algorithms for Distributed systems}, pages 24--29, 2022.

\bibitem{li2022truck}
Hongqi Li, Jun Chen, Feilong Wang, and Yibin Zhao.
\newblock Truck and drone routing problem with synchronization on arcs.
\newblock {\em Naval Research Logistics (NRL)}, 69(6):884--901, 2022.

\bibitem{macrina2020drone}
Giusy Macrina, Luigi Di~Puglia Pugliese, Francesca Guerriero, and Gilbert
  Laporte.
\newblock Drone-aided routing: A literature review.
\newblock {\em Transportation Research Part C: Emerging Technologies},
  120:102762, 2020.

\bibitem{masone2022multivisit}
Adriano Masone, Stefan Poikonen, and Bruce~L Golden.
\newblock The multivisit drone routing problem with edge launches: An iterative
  approach with discrete and continuous improvements.
\newblock {\em Networks}, 80(2):193--215, 2022.

\bibitem{mathew2015optimal}
Neil Mathew, Stephen~L Smith, and Steven~L Waslander.
\newblock Optimal path planning in cooperative heterogeneous multi-robot
  delivery systems.
\newblock In {\em Algorithmic Foundations of Robotics XI: Selected
  Contributions of the Eleventh International Workshop on the Algorithmic
  Foundations of Robotics}, pages 407--423. Springer, 2015.

\bibitem{murray2015flying}
Chase~C Murray and Amanda~G Chu.
\newblock The flying sidekick traveling salesman problem: Optimization of
  drone-assisted parcel delivery.
\newblock {\em Transportation Research Part C: Emerging Technologies},
  54:86--109, 2015.

\bibitem{Raghu2023}
Aishwarya Raghunatha, Emma Lindkvist, Patrik Thollander, Erika Hansson, and
  Greta Jonsson.
\newblock Critical assessment of emissions, costs, and time for last-mile goods
  delivery by drones versus trucks.
\newblock {\em Scientific Reports}, 13(1):11814, 2023.

\bibitem{thomas2023collaborative}
Teena Thomas, Sharan Srinivas, and Chandrasekharan Rajendran.
\newblock Collaborative truck multi-drone delivery system considering drone
  scheduling and en route operations.
\newblock {\em Annals of Operations Research}, pages 1--47, 2023.

\bibitem{dr3}
Li. X., J.~Tupayachi, A.~Sharmin, and M.~Ferguson.
\newblock Drone-aided delivery methods, challenge, and the future: A
  methodological review.
\newblock {\em Drones}, 7:191, 2023.

\bibitem{zhang2023review}
Ruowei Zhang, Lihua Dou, Bin Xin, Chen Chen, Fang Deng, and Jie Chen.
\newblock A review on the truck and drone cooperative delivery problem.
\newblock {\em Unmanned Systems}, pages 1--25, 2023.

\end{thebibliography}

\end {document}